\documentclass[aps, pra, twocolumn, superscriptaddress]{revtex4-2}
\usepackage{amsthm}
\usepackage{amsmath,bm}
\usepackage{amssymb}
\usepackage{amsfonts}
\usepackage{graphicx}
\usepackage{txfonts}
\usepackage{xcolor}
\usepackage{braket}
\usepackage{bbm}
\usepackage{subcaption}
\usepackage{changes}
\usepackage{ragged2e}
\usepackage[colorlinks=true,linkcolor=blue,citecolor=blue,urlcolor=blue]{hyperref}
\usepackage[capitalise]{cleveref}
\usepackage{booktabs}
\usepackage{diagbox}

\newcommand{\ignore}[1]{} 

\newcounter{SaveEqnCntr}

\newcommand{\be}{\begin{equation}}
\newcommand{\ee}{\end{equation}}
\newcommand{\ba}{\begin{eqnarray}}
\newcommand{\ea}{\end{eqnarray}}

\def\>{\rangle}
\def\<{\langle}

\newcommand{\Tr}{\text{Tr}}

\newtheorem{theorem}{Theorem}
\newtheorem{corollary}{Corollary}
\newtheorem{definition}{Definition}

\newtheorem{example}{Example}

\begin{document}
	
\title{Probing Kirkwood-Dirac nonpositivity and its operational implications via moments}	

\author{Sudip Chakrabarty}
\email{sudip27042000@gmail.com}
\affiliation{S. N. Bose National Centre for Basic Sciences, Block JD, Sector III, Salt Lake, Kolkata 700 106, India}

\author{Bivas Mallick}
\email{bivasqic@gmail.com}
\affiliation{S. N. Bose National Centre for Basic Sciences, Block JD, Sector III, Salt Lake, Kolkata 700 106, India}

\author{Saheli Mukherjee}
\email{mukherjeesaheli95@gmail.com}
\affiliation{S. N. Bose National Centre for Basic Sciences, Block JD, Sector III, Salt Lake, Kolkata 700 106, India}

\author{Ananda G. Maity}
\email{anandamaity289@gmail.com}
\affiliation{School of Physical Sciences, Indian Institute of Technology Goa, Ponda 403401, Goa, India}
	
\begin{abstract}
The Kirkwood-Dirac (KD) distribution has recently emerged as a powerful quasiprobability framework with wide-ranging applications in quantum information processing tasks. In this work, we introduce an experimentally motivated criterion for detecting nonclassical signatures of the KD distribution using its statistical moments and demonstrate its effectiveness through explicit examples. We further show that this approach extends naturally to identify other quantum resources, such as quantum coherence and nonclassical extractable work, that are intrinsically connected to the KD distribution. Finally, we propose a methodology to realize these moments using shadow tomography technique. Our criteria involves the evaluation of simple functionals, making it well-suited for efficient experimental implementation.
\end{abstract}
\maketitle
\section{Introduction}
Nonclassicality lies at the heart of quantum mechanics, encapsulating features of quantum systems that defy classical description. Some of the promising nonclassical features for example, superposition \cite{szabo2003quantum}, entanglement \cite{horodecki2009quantum}, nonlocality \cite{RevModPhys.86.419}, and contextuality \cite{PhysRevA.71.052108,budroni2022kochen} underpin the distinct capabilities of quantum systems in areas such as computation, communication, and precision measurement. Yet, a primary question in the study of nonclassicality is to identify and characterize the fundamental differences in how quantum state preparation, evolution, and measurement are represented and interpreted, particularly in ways that reveal their deviation from classical probabilistic frameworks \cite{lee1991measure,mari2011directly,vogel2014unified,innocenti2022nonclassicality,mallick2025efficient}.

A compelling strategy for capturing nonclassical aspects of quantum systems is through the use of quasiprobability representations---mathematical frameworks that extend classical probability theory into the quantum domain \cite{cahill1969density,ferrie2011quasi}. However, a fundamental distinction between these quasiprobability distributions and the classical probability distribution is that they can take negative or even complex values \cite{veitch2012negative,tan2020negativity}. Prominent examples of such distributions that violate certain Kolmogorov's axioms of joint probability function but reflect essential features of quantum behavior include the Wigner distribution \cite{wigner1932quantum,kenfack2004negativity}, the Glauber-Sudarshan P distribution \cite{sudarshan1963equivalence,glauber1963coherent}, and the Husimi Q distribution \cite{husimi1940some}.

Such nonpositivity in quasiprobability distributions often stems from the incompatibility of observables, which precludes assigning joint probabilities to their outcomes. A classic example is the phase-space formulation: while classical systems admit consistent joint probability distributions over phase-space, quantum states admit a Wigner function representation which can take negative values, indicating the absence of a classical probabilistic interpretation \cite{hudson1974wigner}. Recent developments have shown that quantum states exhibiting Wigner negativity can provide a significant advantage in various quantum information processing tasks \cite{adhikari2008teleportation,adhikari2008broadcasting,spekkens2008negativity,niset2009no}. Consequently, Wigner negativity is not merely a signature of nonclassicality but also constitutes a crucial resource for realizing quantum computational advantage \cite{mari2012positive}.

On the other hand, recent developments in quantum information processing mainly focus on discrete-variable systems and a broad range of observables beyond just position and momentum. The well-known Wigner function is however formulated with respect to a fixed pair of conjugate observables---position and momentum. Later, though discrete versions of the Wigner function have been introduced \cite{bjork2008discrete}, they still rely on two fixed, maximally noncommuting observables. To overcome this issue and enable more general measurement scenarios, the Kirkwood-Dirac (KD) distribution has emerged as a more flexible and broadly applicable quasiprobability framework \cite{Budiyono_2023_correlation,budiyono_2024coherence,arvidsson2024properties, Lostaglio2023kirkwooddirac, Budiyono_2024_separation, budiyono_2025_entanglement, burkat_2025_structure}. The KD distribution was introduced independently by Kirkwood in 1933 in the context of thermodynamic partition functions \cite{kirkwood1933quantum}, and later rediscovered by Dirac in 1945, who emphasized its role in bridging classical and quantum mechanics \cite{dirac1945analogy}. Dirac argued that the key distinction between the two lies in the noncommutativity of observables and showed that the KD distribution enables the calculation of expectation values within a quasiprobabilistic framework.

The real part of the KD distribution, rediscovered as the Margenau-Hill quasiprobability (MHQ) distribution in 1961, regained significance with the rise of quantum information theory due to its broad applicability \cite{Margenau61}. Nonpositivity of the KD distribution has been identified as a stronger marker of nonclassicality than noncommutativity, with implications in contextuality \cite{pusey2014anomalous,kunjwal2019anomalous}, measurement disturbance \cite{dressel2012significance,monroe2021weak}, violation of Leggett-Garg inequalities \cite{dressel2011experimental}, etc. In quantum metrology, non-real KD quasiprobabilities enable extracting information beyond classical limits, as seen in weak-value amplification, which boosts the signal-to-noise ratios \cite{ArvidssonShukur2020}. Moreover, KD quasiprobabilities offer an efficient alternative to the tomographic reconstruction of quantum states. The MHQ distribution has been seen to be promising in extending fluctuation theorems in quantum thermodynamics, providing insights into heat-flow anomalies and their connection to nonclassicality \cite{PRXQuantum.5.030201}. Recently, quasiprobability-based approaches have been applied to study work statistics in models such as the quantum Ising chain \cite{gonzalez2019out}.

Given the growing relevance of the KD distribution in various physical contexts, a natural next step is to rigorously characterize and detect its nonclassical features before employing it in information processing or other physically motivated applications. KD nonclassicality is commonly associated with the appearance of negative or imaginary values in the quasiprobability distribution. To quantify this behavior, Alonso {\it et al.} proposed a dedicated measure in \cite{gonzalez2019out} which has been widely studied and applied to detect deviations from classical probability distributions \cite{arvidsson2021conditions,PhysRevLett.127.190404,de2023relating}. Moving beyond this conventional approach, the authors in \cite{PhysRevA.109.012215} introduced a more refined measure of KD nonclassicality based on the quantum modification terms in the KD decomposition proposed by Johansen in \cite{PhysRevA.76.012119}. Other common ways of detection include the weak two-point measurement   (WTPM) \cite{PhysRevA.76.012119}, an interferometric approach to reconstruct the KD distribution \cite{PhysRevLett.110.230602, PhysRevLett.110.230601}, direct reconstruction schemes \cite{lundeen2011direct,PhysRevLett.108.070402}, cloning-based approaches \cite{buscemi2013direct}, witness-based detection schemes \cite{langrenez2024convex}, {\it etc.} Some of these detection schemes have also been implemented in experiments \cite{PhysRevLett.113.140601,PhysRevLett.119.050405,wagner2024quantum,hernandez2024interferometry}. But, the direct schemes involve reconstructing the quasiprobability distribution, which is costly in terms of the resources consumed. Cloning-based approaches are constrained by the no-cloning theorem, and hence they require prior information about the state.

Motivated by the diverse aspects and applications of the KD distribution, in this work, we propose a practical method to detect the nonpositivity of the KD distribution using moment-based criteria, which avoid the need for full process tomography and are thus well-suited for experimental implementation. Nevertheless, as a direct operational consequence, our detection framework also enables the identification of coherence and thermodynamic quantities such as nonclassical extractable work. Specifically, we define the $n$th-order moments of the KD distribution, and show that they can be accessed in real experiments through simple functionals using shadow tomography. Such moment-based techniques have been effectively applied in detecting quantum resources such as entanglement \cite{PhysRevLett.125.200501}, higher Schmidt number \cite{mallick2025higher}, many-body correlations \cite{PhysRevLett.109.130502}, genuine multipartite entanglement \cite{mukherjee2025efficient}, non-Markovianity \cite{PhysRevA.109.022247}, and Wigner negativity \cite{mallick2025efficient}. Moreover, unlike witness-based approaches \cite{langrenez2024convex}, this criterion is state-independent and requires no prior knowledge of the quantum state. We demonstrate its effectiveness through examples involving nonpositive KD distributions.

The rest of the paper is organized as follows. In the next section (Sec.~\ref{s2}), we introduce the detailed overview of the KD distribution, namely its definition and mathematical structure. This is followed by an introduction to the moment criteria and its applications in entanglement theory. In Sec.~\ref{s3}, we propose our criteria for KD nonpositivity detection based on moments. We also provide two explicit examples in support of our criteria. We further demonstrate the operational relevance of our criteria in the contexts of detecting quantum coherence and nonclassical extractable work in Sec.~\ref{s4}. A proposal for evaluating the KD moments using shadow tomography is presented in Sec. \ref{s5}. Finally, in Sec.~\ref{s6}, we summarize our results along with some potential avenues for future research.

\section{Preliminaries}\label{s2}
\subsection{KD Distribution}\label{s2A}

Let us consider a finite-dimensional Hilbert space of dimension $d$, on which all operators are assumed to act. Let ${\ket{a_i}}$ and ${\ket{f_j}}$ denote two orthonormal bases of this space. In what follows, these bases correspond to the eigenstates of the observables ${A} = \sum_i a_i \ket{a_i}\bra{a_i}$ and ${F} = \sum_j f_j \ket{f_j}\bra{f_j}$, respectively. Given a quantum state $\rho$, one may define its KD distribution, $Q(\rho)$, with elements given by:
\begin{equation} \label{Eq:KD}
Q_{ij}(\rho) \equiv \langle f_j | a_i \rangle \bra{a_i} \rho \ket{f_j} = \mathrm{Tr}(\Pi^f_j \Pi^a_i \rho),
\end{equation}
where $\Pi^a_i = \ket{a_i}\bra{a_i}$ and analogously for $\Pi^f_j$. This distribution can be used to calculate expectation values and measurement-outcome probabilities. The elements
$\{ Q_{ij} \}$ satisfies some of the Kolmogorov's axioms for joint probability distributions \cite{Kolmogorov51}:
\begin{align}
\sum_{ij} Q_{ij}  = 1 , \; \sum_{j} Q_{ij}  = p(a_i|\rho) , \; \mathrm{and} \;
\sum_{i} Q_{ij}  = p(f_j|\rho) , \nonumber
\end{align}
where $p(a_i|\rho)$ and $p(f_j|\rho)$ denote the measurement probabilities associated with the bases ${\ket{a_i}}$ and ${\ket{f_j}}$, respectively.

In the special case where the bases coincide, {\it i.e.}, ${\ket{a_i}} = {\ket{f_j}}$, the KD distribution reduces to a classical probability distribution:
\begin{align*}
\{Q_{ij}(\rho)\} &= \{\langle f_j | a_i \rangle \bra{a_i} \rho \ket{f_j} \delta_{f_j, a_i} \} = \{ \mathrm{Tr}(\Pi^a_i \rho) \delta_{f_j, a_i} \}.
\end{align*}

When all observables commute the KD distribution reduces to a valid probability distribution. In contrast, for general quantum systems, $Q_{ij}(\rho)$ can take negative or even complex values, indicating nonclassical behavior.
Certain physical processes \cite{Yunger18, gonzalez2019out, Razieh19, ArvidssonShukur2020} motivate the extension of the KD distribution from $2$ to $k$ bases, corresponding to $k$ observables ${A}^{(1)}, \ldots, {A}^{(k)}$. If ${A}^{(r)} = \sum_j a^{(r)}_j \Pi^{a^{(r)}}_j$ is the  spectral decomposition of ${A}^{(r)}$, then the extended KD distribution is defined as:
\begin{align}
Q^\star_{i_1,\ldots,i_k} (\rho) \equiv \mathrm{Tr}(\Pi^{a^{(k)}}_{i_k} \cdots \Pi^{a^{(1)}}_{i_1}\, \rho).
\end{align} \label{KDExt}
A more detailed discussion and application of the extended KD distribution will be presented in Sec.~\ref{s4}.

It is worth noting that the elements of the KD distribution serve as the coefficients in an operator expansion of the density matrix, $\rho$:
\begin{equation}
\rho = \sum_{i_1, \dots, i_k} \frac{\ket{a_{i_1}^{(1)}}\bra{a_{i_k}^{(k)}}}{\langle a_{i_k}^{(k)} | a_{i_1}^{(1)} \rangle} Q^\star_{i_1, \dots , i_k} (\rho).
\end{equation} \label{KDDecom}

Another operational interpretation of the KD distribution emerges in the context of weak measurement theory. The expectation value of an observable $\mathcal{O}$ may be expressed as:
\begin{equation}
\langle \mathcal{O} \rangle = \mathrm{Tr}(\mathcal{O} \rho) = \sum_{i,j} Q_{ij} \mathcal{O}^w_{ij},
\end{equation}
where $\mathcal{O}^w_{ij} = \bra{f_j} \mathcal{O} \ket{a_i} / \langle f_j | a_i \rangle$ represents the weak value of $\mathcal{O}$, conditioned on the initial state $\ket{a_i}$ and final state $\ket{f_j}$ \cite{Vaidman88, Duck89}.

In general, KD distribution is not a true probability distribution, as its elements $Q_{ij}$ may be negative or even complex. If one restricts only to the real parts of $Q_{ij}$, the corresponding distribution is called the MHQ distribution \cite{Margenau61}, which has found several applications in quantum thermodynamics \cite{PRXQuantum.1.010309,PhysRevResearch.6.023280}. Formally, the elements of the MHQ distribution are given by
\begin{equation}
    Q^{\mathrm{MHQ}}_{ij}(\rho) = \mathrm{Re}\!\left[ Q_{ij}(\rho) \right]
    = \mathrm{Re}\!\left[\langle f_j | a_i \rangle \langle a_i | \rho | f_j \rangle \right]. \label{mhq}
\end{equation}
$Q^{\mathrm{MHQ}}_{ij}(\rho)$ provides a real-valued normalized quasiprobability representation of the state $\rho$, where negative values signal nonclassicality associated with the bases $\{ |a_i\rangle \}$ and $\{ |f_j\rangle \}$.

The KD distribution provides a framework for bounding values attainable in classical scenarios and has found broad applications in quantum information theory \cite{Dressel15, Yunger18,gonzalez2019out, Razieh19,ArvidssonShukur2020, kunjwal2019anomalous, pusey2014anomalous, dressel2011experimental,lostaglio2020certifying,Upadhyaya24}. For a comprehensive overview, we refer interested readers to Ref. \cite{arvidsson2024properties}.

\subsection{Moment criteria}

The use of moments of the partially transposed density matrix---referred to as PT moments---was originally proposed by Calabrese {\it et al.} to study quantum correlations in many-body systems, particularly in the context of relativistic quantum field theory \cite{PhysRevLett.109.130502}. These moments offer a pathway to indirectly access the eigenvalues \(\{\lambda_i\}\) of the partially transposed matrix \(\rho^{T_B}_{AB}\). The corresponding characteristic polynomial is
\begin{equation}
    \text{Det}(\rho^{T_B}_{AB} - \lambda I) = \sum_k a_k \lambda^k,
\end{equation}
where the coefficients \(a_k\) are functions of the PT moments, defined as
\begin{equation}
    p_n = \text{Tr}\left[(\rho_{AB}^{T_B})^n\right], \label{PTmoments}
\end{equation}
for integers \(n \geq 1\). These moments offer a pathway to find the spectrum of $\rho^{T_B}_{AB}$ in an experiment friendly way, which is otherwise impossible due to the unphysical nature of the transposition map.
By construction, \(p_1 = \text{Tr}[\rho_{AB}^{T_B}] = 1\), while the second moment \(p_2\) is associated with the purity of the partially transposed state. The third moment \(p_3\), on the other hand, captures subtler features beyond purity and is used for entanglement detection. Specifically, it was shown in Ref.~\cite{PhysRevLett.125.200501} that for any PPT (positive partial transpose) state
\begin{equation}
    p_2^2 \leq p_3.
\end{equation}
Violation of this condition implies the presence of entanglement, giving rise to the so-called \(p_3\)-PPT criterion. In the special case of Werner states, this criterion is both necessary and sufficient, and hence equivalent to the traditional PPT condition.

Beyond \(p_3\), the higher-order PT moments (\(n \geq 4\)) offer an independent criterion for entanglement detection. To systematically exploit this, the work in Ref.~\cite{yu2021optimal} introduced the concept of Hankel matrices constructed from the sequence of moments. Given a vector of moments \(\mathbf{p} = (p_1, p_2, \ldots, p_{2m+1})\), the \( (m+1) \times (m+1) \) Hankel matrix \( H_m(\mathbf{p}) \) is defined element wise as
\begin{equation}
    [H_m(\mathbf{p})]_{ij} := p_{i+j+1}, \quad i,j \in \{0, 1, \ldots, m\}.
    \label{Hankelmatrices}
\end{equation}
Hence, the first and second order matrices are defined as
\begin{equation}
    H_1 = \begin{pmatrix}
p_1 & p_2 \\
p_2 & p_3
\end{pmatrix} \quad \text{and} \quad
H_2 = \begin{pmatrix}
p_1 & p_2 & p_3 \\
p_2 & p_3 & p_4 \\
p_3 & p_4 & p_5
\end{pmatrix}.
\end{equation}
Using these Hankel matrices, a necessary condition for separability is
\begin{equation}
    \det[H_m(\mathbf{p})] \geq 0.
    \label{Hankelmatrixcondition}
\end{equation}

A practical advantage of PT moments is that they can be estimated experimentally via shadow tomography techniques~\cite{aaronson2018shadow,huang2020predicting,cieslinski2024analysing}. These methods circumvent the need for full state tomography, significantly lowering the experimental overhead. Since moment-based criteria rely only on a polylogarithmic number of state copies (as opposed to the exponential scaling required for full tomography), they present an efficient and scalable route to entanglement detection, particularly in high-dimensional and many-body systems. Furthermore, the criterion being state-independent, is suitable for applications where little or no prior information about the system is available.

Motivated by the above considerations, in the next section we explore how moment-based detection schemes can be developed for the detection of nonpositivity of the KD distribution. Firstly, we define the moments of KD distribution $(q_n)$, and based on it we develop our formalism for the detection of KD nonpositivity.

\section{Detection of KD nonpositivity} \label{s3}
In this section, we present our main results for detecting KD nonpositivity using moments of the KD distribution. We also provide explicit examples that support our proposed theorem, aiming to firmly establish its validity.\\

\begin{definition} 
Let ${\rho}$ be an arbitrary density operator belonging to the set of bounded operators $\mathcal{B}(\mathbb{H})$ acting on a $d$-dimensional Hilbert space $\mathbb{H}$, and $Q(\rho)$ be the KD distribution for the quantum state $\rho$ with respect to bases $\{\ket{a_i}\}$ and $\{\ket{f_j}\}$ as defined in Sec.~\ref{s2A}. We define the $n$-th order KD moments $(q_n) $ as:
\begin{equation}
\begin{split}
 q_n := \sum_{i,j}
 (Q_{ij} (\rho))^n\label{kdmoments}
\end{split}
\end{equation}
where $n$ is an integer.
\end{definition}

For brevity, we denote the KD distribution by \( Q \equiv Q(\rho) \), with elements \( Q_{ij} \equiv Q_{ij}(\rho) \), whenever the state \( \rho \) is prominent from the context. Following the definition of KD moments in Eq.~\eqref{kdmoments}, we now propose our criterion for detecting KD nonpositivity in the following theorems.\\

\begin{theorem} \label{theorem1}
If the KD distribution $(Q (\rho) )$ for a quantum state $\rho$ with respect to bases $\{\ket{a_i}\}$ and $\{\ket{f_j}\}$ is positive, then 
\begin{equation}
{q_2}^2 \leq q_3 \label{q2_criterion}
\end{equation}
where $q_2$ and $q_3$ are defined in \eqref{kdmoments}.
\end{theorem}

\begin{proof}
Let the KD distribution $(Q )$ be a positive, real-valued, and normalized distribution. Let us now consider the vector $l_{p}$ norm of $Q$ for $p \ge 1$, defined as:
 \begin{equation}
 \begin{split}
   & ||Q||_{l_P} := \left( \sum_{i,j} |Q_{ij}| ^{p} \right)^{\frac{1}{p}} .\label{schattenp}
    \end{split} 
\end{equation}
\vspace{0.1cm}

The standard inner product corresponding to a $d$ vector is defined as
 \begin{equation}
     \langle u, v \rangle :=\sum_{i=1}^{d} u_{i} v_{i} \label{innerproduct}
 \end{equation}
 for $u,v \in \mathbb{R}^{d}$.

\textit{H\"{o}lder's inequality for vector $l_{p}$ norm.} For $p, q \ge 1$ and $\frac{1}{p}+\frac{1}{q} =1$, the following relation holds:
 \begin{equation}
     |\langle u, v\rangle | \le \sum_{i=1}^{d} |u_{i} v_{i}| \le ||u||_{l_p} ||v||_{l_q} .\label{hoelderinequality}
 \end{equation}
Note that, for $p=q=2$, H\"{o}lder's inequality reduces to the Cauchy-Schwarz inequality.\\

In our case, setting $p=3$, $q=\frac{3}{2}$ in H\"{o}lder's inequality and choosing $u=v=Q$, we obtain
\begin{equation}
     ||Q.Q||_{l_1} \leq  ||Q||_{l_3} ||Q||_{{l}_{\frac{3}{2}}}. \label{hoelder}
     \end{equation}
Now, 
\begin{equation}
\begin{split}
&||Q||_{l_{\frac{3}{2}}}  =\left(\sum_{i,j} \hspace{1mm} |Q_{ij}|^{\frac{3}{2}} \right)^{\frac{2}{3}} \\
& = \left(\sum_{i,j} \hspace{1mm} |Q_{ij}|\cdot |Q_{ij}|^{\frac{1}{2}}\right)^{\frac{2}{3}}\\
  & \leq \left[\left(\sum_{i,j} \hspace{1mm} |Q_{ij}|^{2} \hspace{1mm}   \right)^{\frac{1}{2}}  \left(\sum_{i,j} \hspace{1mm} |Q_{ij}|  \right)^{\frac{1}{2}}\right]^{\frac{2}{3}}  \\
   & = \left[\left(\sum_{i,j} \hspace{1mm} |Q_{ij}|^{2} \right)^{\frac{1}{3}} \left(\sum_{i,j} \hspace{1mm} |Q_{ij}| \right)^{\frac{1}{3}}\right]
\end{split}
\end{equation}
{\it i.e.},
\begin{equation}
||Q||_{l_{\frac{3}{2}}} \leq  {||Q||^{\frac{2}{3}}_{l_2}} {||Q||^{\frac{1}{3}}_{l_1}} .\label{cs}
\end{equation}
Putting the value of Eq.~\eqref{cs} in Eq.~\eqref{hoelder}, we get,
\begin{equation}
     ||Q.Q||_{l_{1}} \leq  ||Q||_{l_{3}}   {||Q||^{\frac{2}{3}}_{l_2}}    {||Q||^{\frac{1}{3}}_{l_1}}. \label{16}
\end{equation}
Further note that, 
\begin{equation}
\begin{split}
||Q.Q||_{l_1}  = {||Q||^{2}_{l_2}}.\\ 
\end{split} \nonumber
\end{equation}
Therefore, it follows that
     \begin{equation}
     {||Q||^2_{l_{2}}} \leq  ||Q||_{l_{3}}{||Q||^{\frac{2}{3}}_{l_2}} ||Q||^{\frac{1}{3}}_{l_1}. \label{16} \nonumber
     \end{equation}
     Taking $3$rd power on both sides, we get
     \begin{equation}
     {||Q||^6_{l_2}} \leq  {||Q||^{3}_{l_{3}}}{||Q||^{2}_{l_2}} {||Q||^{1}_{l_1}} \label{17} \nonumber
     \end{equation}
     {\it i.e.},  \begin{equation}
     {||Q||^4_{l_{2}}} \leq  {||Q||^{3}_{l_{3}}} {||Q||^{1}_{l_1}} .  \label{main}
     \end{equation}
On the other hand, from the normalization condition of the KD distribution $(Q)$, we get 
\begin{equation}
       ||Q||_{l_1}  = \left(\sum_{i,j} \hspace{1mm} |Q_{ij}|^{1} \hspace{1mm}  \right)^1
    =1. \nonumber
\end{equation}
Since we have assumed that the KD distribution is positive, {\it i.e.} $Q_{i,j} = |Q_{ij}|, \hspace{1mm} \forall  \hspace{1mm}i,j$, 
Therefore, from Eq. \eqref{main}, we get
\begin{equation}
     {||Q||^4_{l_{2}}} \leq  {||Q||^{3}_{l_{3}}}, \nonumber
     \end{equation} 
     or 
     \begin{equation}
        {q_2}^2 \leq q_3,
\end{equation}
which completes our proof.
\end{proof}

Above theorem implies that condition presented in Eq.~\eqref{q2_criterion} is necessary for positivity of a KD distribution. Hence, violation of the above theorem is sufficient to conclude that the KD distribution is nonpositive. Below we present an example in support of Theorem~\ref{theorem1}.

\begin{example} \label{Ex:Saturated Nonclassicality}
Suppose that ${A}$ and ${F}$ act on a two-qubit Hilbert space and have eigenbases 
\begin{align}
    \{ \ket{a_i} \} &= \{ \ket{0}\ket{0}, \ket{0}\ket{1}, \ket{1}\ket{0}, \ket{1}\ket{1} \}\\
    \{ \ket{f_j} \} &= \{ \ket{+}\ket{+}, \ket{-}\ket{+}, \ket{+}\ket{-}, \ket{-}\ket{-} \}. 
\end{align} 
Note that $\{ \ket{a_i} \}$ and $\{ \ket{f_j} \}$ form a pair of mutually unbiased bases (MUBs). Consider the state,
\begin{align}
    {\rho} = p \ket{\Psi}\bra{\Psi} +\frac{1-p}{4}\mathbb{I} ,
\end{align}  
where $\ket{\Psi} = (\ket{0}\ket{0}+\ket{0}\ket{1}+\ket{1}\ket{0}-\ket{1}\ket{1})/2$. The overlaps $\left| \langle \Psi \mid a_i \rangle \right|
 = \left| \langle \Psi \mid f_j \rangle \right|
 =\left| \langle a_i \mid f_j \rangle \right|
 = \frac{1}{2}$ for all $i,j$. The resulting KD distribution is given in Table \ref{Tab:Saturated Nonclassicality}.
\end{example}

\begin{table}[t]
\caption{\label{Tab:Saturated Nonclassicality}
KD distribution corresponding to Example~\ref{Ex:Saturated Nonclassicality}.
}
\begin{ruledtabular}
\begin{tabular}{ccccc}
\diagbox{$\ket{f_j}$}{$\ket{a_i}$} & $\ket{0}\ket{0}$ & $\ket{0}\ket{1}$ & $\ket{1}\ket{0}$ & $\ket{1}\ket{1}$ \\ 
\midrule
$\ket{+}\ket{+}$ &
$\frac{1+p}{16}$ &
$\frac{1+p}{16}$ &
$\frac{1+p}{16}$ &
$\frac{1-3p}{16}$ \\ \\

$\ket{-}\ket{+}$ &
$\frac{1+p}{16}$ &
$\frac{1+p}{16}$ &
$\frac{1-3p}{16}$ &
$\frac{1+p}{16}$ \\ \\

$\ket{+}\ket{-}$ &
$\frac{1+p}{16}$ &
$\frac{1-3p}{16}$ &
$\frac{1+p}{16}$ &
$\frac{1+p}{16}$ \\ \\

$\ket{-}\ket{-}$ &
$\frac{1-3p}{16}$ &
$\frac{1+p}{16}$ &
$\frac{1+p}{16}$ &
$\frac{1+p}{16}$ \\
\end{tabular}
\end{ruledtabular}
\end{table}
It can be readily verified that the KD distribution becomes nonpositive whenever $p>1/3$. Using Eq. \eqref{kdmoments}, we calculate $q_2^2-q_3$ for this example and get,
\begin{align}
    q_2^2 - q_3 = \frac{9p^4}{256} + \frac{3p^3}{128} -\frac{3p^2}{256}.
\end{align}
Now, from Theorem \ref{theorem1}, we get the KD distribution is nonpositive when 
\begin{align}
    q_2^2 -q_3 > 0 \;\; \Rightarrow p > \frac{1}{3}.
\end{align}
This validates Theorem 1.

We now present a second theorem, which further strengthens Theorem~\ref{theorem1} by incorporating higher-order moments.\\

\begin{theorem}\label{theorem2}
    If the KD distribution $(Q (\rho))$ of a quantum state $\rho$ with respect to bases $\{\ket{a_l}\}$ and $\{\ket{f_p}\}$ with $l,p =1,2,\dots,d$, is positive, then

\begin{equation}
       \det\left[H_{m}(\mathbf{q})\right] \ge 0 . \label{Hankelmatrix}
   \end{equation}
Here, $[H_{m}(\mathbf{q})]_{ij} = q_{i+j+1}$ for $i,j \in \{0,1,\dots,m\}$, $m \in \mathbb{N}$ and $ \mathbf{q}= (q_1, q_2, \dots, q_{2m+1})$ are the corresponding KD moments defined in Eq.~\eqref{kdmoments}.
\end{theorem}

\begin{proof}
We consider the KD distribution $(Q)$ for a quantum state $\rho$ in a $d$-dimensional Hilbert space, defined with respect to two bases $\{ \ket{a_l} \}$ and $\{ \ket{f_p} \}$, with $l,p = 1,2,\dots,d$, to be a real-valued, positive, normalized distribution. Since there are $d$ choices for both $l$ and $p$, the distribution consists of $d^2$ elements in total, which is denoted as $\{ Q_{ij} \}$.

From these $d^2$ elements we construct a diagonal matrix $D$ of size $d^2 \times d^2$ by placing each KD element on the diagonal. Explicitly, 
\[
    D = \mathrm{diag}(Q_{11}, Q_{12}, \dots, Q_{1d}, Q_{21}, \dots, Q_{dd}).
\]
The order in which the KD elements are arranged on the diagonal is immaterial for the argument. For convenience, we label the diagonal entries as $\chi_k$, so that
\[
    D = \sum_{k=1}^{d^2} \chi_k \ket{k}\bra{k}, \qquad \chi_k \geq 0, \;\; \sum_{k=1}^{d^2} \chi_k = 1.
\]
By construction, $D$ is a diagonal, positive semidefinite matrix with unit trace, and it is already written in spectral decomposed form.  

Now, consider the KD moment vector $\mathbf{q} = (q_1, q_2, \dots, q_{2m+1})$ as introduced in Theorem~\ref{theorem2}. The associated Hankel matrix $H_m(\mathbf{q})$, constructed from these moments, admits a Vandermonde decomposition of the form~\cite{boley1997vandermonde,tyrtyshnikov1994bad,heinig2013algebraic},
\begin{equation}
    H_m(\mathbf{q}) = V_m D V_m^T,
\end{equation}
where $V_m$ is the associated Vandermonde matrix:

\begin{equation}
V_m = 
\begin{pmatrix}
1 & 1 & \cdots &1  \\ 
\chi_1 & \chi_2 & \cdots & \chi_{d^2}  \\
\vdots & \vdots & \ddots & \vdots \\
\chi_1^m & \chi_2^m & \cdots & \chi_{d^2}^m 
\end{pmatrix}
\end{equation}
and
\begin{equation}
D = 
\begin{pmatrix}
 \chi_1& 0 & \cdots &0  \\ 
0 & \chi_2 & \cdots & 0 \\
\vdots & \vdots & \ddots &\vdots\\
0 & 0 & \cdots &   \chi_{d^2}
\end{pmatrix}
\end{equation}
Now, for an arbitrary vector $ \boldsymbol{x}=(x_1, \dots ,x_m, x_{m+1}) \in {\mathbb{R}}^{m+1}$, we have
\begin{align}
     \boldsymbol{x}  H_{m}(\mathbf{q}) \boldsymbol{x}^T &= \boldsymbol{x}  V_m D V_m^T  \boldsymbol{x}^T \nonumber \\
     &= \boldsymbol{y} D \boldsymbol{y}^T \nonumber \\
     &= \sum_{i=1}^{d^2} \chi_i {y_i}^2 \ge 0,
\end{align}
where, $\boldsymbol{y}= \boldsymbol{x}  V_m = (y_1, y_2,\dots, y_{d^2}) $ with $y_i= \sum_{j=1}^{m+1} x_j \hspace{0.1cm}{\chi_i}^{j-1}, \hspace{0.1cm} \text{ for } i = 1, 2, \dots, d^2$.

Hence, \( \boldsymbol{x} H_{m}(\mathbf{q}) \boldsymbol{x}^T \ge 0 \), which implies that \( H_{m}(\mathbf{q}) \) is positive semidefinite, i.e., \( H_{m}(\mathbf{q}) \ge 0 \). Consequently, its determinant satisfies \( \det[H_{m}(\mathbf{q})] \ge 0 \), thereby completing the proof.
\end{proof}

Similar to Theorem~\ref{theorem1}, Theorem~\ref{theorem2} establishes that the condition given in Eq.~\eqref{Hankelmatrix} is a necessary requirement for the positivity of the KD distribution. Consequently, any violation of this condition is sufficient to conclude that the KD distribution is nonpositive. It is important to note that the lowest-order case of Theorem~\ref{theorem2}, corresponding to \(m=1\), yields \(\det[H_1(\mathbf{q})] = q_3 - q_2^2 \geq 0\), which precisely recovers the criterion presented in Theorem~\ref{theorem1}.  
For \(m>1\), Theorem~\ref{theorem2} provides a hierarchy of progressively stronger criteria to detect the nonpositivity of the KD distribution.  
We now present an example to validate Theorem~\ref{theorem2} and demonstrate how it offers a stronger detection condition than Theorem~\ref{theorem1}.

\begin{example}\label{Ex:NonClKDClassicalIneq}
Consider a two-qubit system. We choose ${A}$ and ${F}$ such that  $\{ \ket{a_i} \} = \{ \ket{0}\ket{0}, \ket{0}\ket{1}, \ket{1}\ket{0}, \ket{1}\ket{1} \}$ and $\{ \ket{f_j} \} = \{ \ket{0}\ket{+}, \ket{0}\ket{-}, \ket{1}\ket{+}, \ket{1}\ket{-} \}$. 
We set ${\rho} = \ket{\Psi} \bra{\Psi}$, where $\ket{\Psi} = \left( \ket{0}\ket{0}+2\ket{0}\ket{1} \right)/\sqrt{5}$. The corresponding KD distribution is presented in Table \ref{Tab:NonClKDClassicalIneq}.
\end{example}

\begin{table}[t]
\caption{\label{Tab:NonClKDClassicalIneq}
KD distribution corresponding to Example~2.
}
\begin{ruledtabular}
\begin{tabular}{ccccc}
\diagbox{$\ket{f_j}$}{$\ket{a_i}$} & $\ket{0}\ket{0}$ & $\ket{0}\ket{1}$ & $\ket{1}\ket{0}$ & $\ket{1}\ket{1}$ \\ 
\midrule
$\ket{0}\ket{+}$ & $\frac{3}{10}$ & $\frac{3}{5}$ & $0$ & $0$ \\ \\
$\ket{0}\ket{-}$ & $-\frac{1}{10}$ & $\frac{1}{5}$ & $0$ & $0$ \\ \\
$\ket{1}\ket{+}$ & $0$ & $0$ & $0$ & $0$ \\ \\
$\ket{1}\ket{-}$ & $0$ & $0$ & $0$ & $0$ \\ 
\end{tabular}
\end{ruledtabular}
\end{table}

Using \eqref{kdmoments}, we calculate \( (q_2^2 - q_3 )\) for this example and find
$q_2^2 - q_3 = 0.$ Thus, according to Theorem~\ref{theorem1}, the nonpositivity of the KD distribution is not detected for this example.  
However, when we apply our proposed criterion in Theorem~\ref{theorem2}, we find \(\det(H_2) = -2.0736 \times 10^{-4} < 0\).  
The negative determinant clearly reveals the nonpositivity of the KD distribution, demonstrating that higher-order moment conditions, as captured by Theorem~\ref{theorem2}, are effective in detecting nonpositivity in this case.

\section{Detecting quantum coherence and nonclassical extractable work as a direct consequence of KD nonpositivity detection}\label{s4}
In this section, we investigate whether our proposed criterion for detecting KD nonpositivity can also be applied to identifying other quantum resources, inherently linked to it. In particular, we demonstrate that, as a direct extension, the criterion can effectively be used to detect quantum resources such as quantum coherence and nonclassical extractable work.

\subsection{Detecting Quantum Coherence}
 Coherence is one of the fundamental elements of quantum theory that arises from the superposition principle, underpinning many nonclassical features \cite{streltsov2017colloquium}. In this subsection, we will discuss how quantum coherence can be detected as a direct implication of our proposed moment criteria.

 Let $\rho$ be a quantum state belonging to the set of bounded operators $\mathcal{B}(\mathbb H)$ acting on a $d$-dimensional Hilbert space ($\mathbb{H}$) and consider an orthonormal basis $\{\ket{a_i}\}$ on $\mathbb{H}$ ({\it i.e.}, $\sum_i \ket{a_i} \bra{a_i} = \mathbb{I}$). The quantum state $\rho$ can be expressed in terms of the orthonormal basis $\{\ket{a_i}\}$ as
 \begin{equation}
     \rho = \sum_{ij} a_{ij} \ket{a_i} \bra{a_j}
 \end{equation}
where, $ a_{ij} = \bra{a_i} \rho \ket{a_j}$. The state $\rho$ is said to be incoherent with respect to an orthonormal basis $\{\ket{a_i}\}$, if it is represented by a diagonal matrix in that basis {\it i.e.}, the off-diagonal elements $a_{ij}=0$ for all $i\neq j$. If, however, there exists at least one non-zero off-diagonal element {\it i.e.}, $a_{ij}\neq 0$ for any $i\neq j$, then $\rho$ is considered coherent with respect to the basis $\{\ket{a_i}\}$. A natural way to quantify quantum coherence is by analyzing the off-diagonal elements of the density matrix $\rho$, as these components directly capture the presence of superposition with respect to the orthonormal basis $\{\ket{a_i}\}$. One widely adopted measure of coherence based on this idea is the $\ell_1$-norm of coherence \cite{baumgratz2014quantifying}, defined as:
\begin{equation} \label{coherencemeasure}
    \mathcal{C}_{\ell_1} (\rho, \{\ket{a_i}\} ) = \sum_{i \neq j} |a_{ij}|.
\end{equation}
 Note that a quantum state $\rho$ can be expressed with respect to more than one set of measurement operators, giving rise to the concept of extended Kirkwood-Dirac (KD) distribution \cite{arvidsson2024properties}. While the general definition is provided in Sec. \ref{s2A}, here we consider two orthonormal bases $\{\ket{a_i}\}$ and $\{\ket{b_k}\}$, where $\{\ket{b_k}\}$ is mutually unbiased basis with respect to $\{\ket{a_i}\}$. In this case, the elements of the extended KD distribution for the quantum state $\rho$ takes the form:
 \begin{equation}
   Q^{\star}_{i,j,k} (\rho) = \langle a_j | b_k \rangle \langle b_k | a_i \rangle \langle a_i | \rho | a_j \rangle .
 \end{equation}

To demonstrate the relevance of our moment-based approach in detecting quantum coherence, we begin by defining the moments of the extended KD distribution.

\begin{definition}
 Let $\mathbb {H}$ be a $d$-dimensional Hilbert space and ${\rho}$ be an arbitrary density operator defined on $\mathbb H$. Consider $Q^{\star}(\rho)$ to be the extended KD distribution for the quantum state $\rho$ with respect to bases $\{\ket{a_i}\}$ and $\{\ket{b_k}\}$, where $\{\ket{b_k}\}$ is mutually unbiased with respect to $\{\ket{a_i}\}$. We define the $n$-th order extended KD moments $(r_n) $ as:

\begin{equation} 
\begin{split}
 r_n := \sum_{i,j,k}
 ( Q^{\star}_{i,j,k} (\rho) )^n \label{extendedkdmoments}
\end{split}
\end{equation}
where $n$ is an integer.
\end{definition}

Based on the above definition, we now formulate a refined detection criterion, using the moments of the extended KD distribution.

\begin{corollary} \label{theorem3}
If the extended KD distribution $Q^{\star}(\rho)$ for a quantum state $\rho$ with respect to bases $\{\ket{a_i}\}$ and $\{\ket{b_k}\}$ is positive, where $\{\ket{b_k}\}$ is mutually unbiased basis with respect to $\{\ket{a_i}\}$, then 
\begin{equation}
\det[H_{m}(\mathbf{r})] \ge 0 . \label{Hankelmatrix_coherence}
\end{equation}
Here, $[H_{m}(\mathbf{r})]_{ij} = r_{i+j+1}$ for $i,j \in \{0,1,\dots,m\}$, $m \in \mathbb{N}$ and $ \mathbf{r}= (r_1, r_2, \dots, r_{2m+1})$ are the corresponding extended KD moments defined in Eq.~\eqref{extendedkdmoments}.
\end{corollary}

\begin{proof}
The argument parallels that of Theorem~\ref{theorem2}, with the key differences being the replacement of the standard KD distribution $Q$ by the extended KD distribution $Q^{\star}$, and the corresponding replacement of the Hankel matrix $H_{m}(\mathbf{q})$ with $H_{m}(\mathbf{r})$.
\end{proof}

Corollary~\ref{theorem3}, much like Theorem~\ref{theorem1} and  Theorem~\ref{theorem2}, asserts that the condition stated in Eq.~\eqref{Hankelmatrix_coherence} is necessary for the extended KD distribution to remain positive. As a result, any violation of this condition serves as a sufficient indicator of nonpositivity. For \(m > 1\), Corollary~\ref{theorem3} introduces a hierarchy of increasingly stronger criteria that can detect nonpositivity with higher sensitivity.

Our goal now is to establish the implications of our proposed moment criteria in detecting quantum coherence, which we formalize in the following theorem.

\begin{theorem}\label{theorem4}
Let $Q^{\star}(\rho)$ be the extended KD distribution for a quantum state $\rho$ defined with respect to bases $\{\ket{a_i}\}$ and $\{\ket{b_k}\}$, where $\{\ket{b_k}\}$ is mutually unbiased basis with respect to $\{\ket{a_i}\}$. If there exists some $m \in \mathbb{N}$ for which the inequality,
\begin{equation}
 - \det[H_{m}(\mathbf{r})] > 0 \label{moment_coherence}, 
\end{equation}
holds, then the state $\rho$ possesses a non-zero coherence in the basis $\{\ket{a_i}\}$ {\it i.e.},
\begin{equation}
    \mathcal{C}_{\ell_1} (\rho, \{\ket{a_i}\} ) >0 \nonumber
\end{equation}
where $\mathcal{C}_{\ell_1} (\rho, \{\ket{a_i}\} )$ is the $\ell_1$ norm of coherence as defined in \eqref{coherencemeasure}. 
\end{theorem}

\proof Consider $Q^{\star}(\rho)$ to be the extended KD distribution for a quantum state $\rho$ defined with respect to bases $\{\ket{a_i}\}$ and $\{\ket{b_k}\}$, where $\{\ket{b_k}\}$ is mutually unbiased with respect to $\{\ket{a_i}\}$. If the inequality $ - \det[H_{m}(\mathbf{r})] > 0$ holds for some $m \in \mathbb{N}$, then by Corollary~\ref{theorem3}, the extended KD distribution $Q^{\star}(\rho)$ must be nonpositive, implying that $ \sum_{i,j,k}|Q_{i,j,k}^{\star}(\rho)| > 1$. Now the $\ell_1$ norm of coherence is given by,
\begin{equation}
\begin{split}
    &  \mathcal{C}_{\ell_1} (\rho, \{\ket{a_i}\} ) = \sum_{i \neq j} |a_{ij}| \\
     &  = \sum_{i,j} |\bra{a_i} \rho \ket{a_j}| -1 \\
     & = \sum_{i,j,k} | \bra{a_j} \ket{b_k} \bra{b_k} \ket{a_i}       \bra{a_i} \rho \ket{a_j}| -1\\
     & = \sum_{i,j,k}|Q_{i,j,k}^{\star}(\rho)| - 1 \\
     & >0.
\end{split}
\end{equation}
This completes our proof. \qed

We now present an example validating Theorem~\ref{theorem4}, demonstrating how quantum coherence can be detected using our proposed moment criteria.

\begin{example}\label{example_coherence}
Let us assume a pure single qubit state, which can be written as 
 \begin{equation}
     \ket{\Psi} = \cos{\frac{\theta}{2}} \ket{0} + \sin{\frac{\theta}{2}} e^{i\alpha} \ket{1}
 \end{equation}
 where $0 \le \theta \le \pi$ and $0 \le \alpha \le 2\pi$.

 The $l_1$ norm of coherence of this pure state with respect to orthonormal basis $ \{\ket{a_i} \}=\{ \ket{0}, \ket{1}\}$ is equal to $|\sin{\theta}|$ \cite{budiyono2023quantifying}. 
 We consider another mutually unbiased basis $\{ \ket{b_k}\}$ with respect to $\{ \ket{a_i}\}$,
 \begin{align}
     \ket{b_1} &= \frac{1}{\sqrt{2}}( \ket{0} + e^{i\beta} \ket{1} ), \nonumber \\
     \ket{b_2} &= \frac{1}{\sqrt{2}}( \ket{0} - e^{i\beta} \ket{1} )
 \end{align}
 with $0 \le \beta\le 2\pi$. The corresponding extended KD distribution is,
 \begin{alignat}{3}
    Q^{\star}_{1,1,1} & =\; Q^{\star}_{1,1,2}  &\quad =\; & \frac{1}{2} \cos^2 \left( \frac{\theta}{2} \right), \nonumber \\
    Q^{\star}_{1,2,1} & =\; - Q^{\star}_{1,2,2} &\quad =\; & \frac{1}{4} e^{-i(\alpha - \beta)} \sin (\theta),\nonumber \\
    Q^{\star}_{2,1,1} & =\; - Q^{\star}_{2,1,2} &\quad =\; & \frac{1}{4} e^{-i(\alpha - \beta)} \sin (\theta), \nonumber\\
    Q^{\star}_{2,2,1} & =\; Q^{\star}_{2,2,2}  &\quad =\; & \frac{1}{2} \sin^2 \left( \frac{\theta}{2} \right). \nonumber
\end{alignat}
\end{example}

From Theorem~\ref{theorem4}, we know that if \( -\det[H_{m}(\mathbf{r})] > 0 \) for some \( m \in \mathbb{N} \), then the state exhibits non-zero coherence in the computational basis \( \{ \ket{0}, \ket{1} \} \), quantified by a non-zero \( l_1 \)-norm of coherence: \( \mathcal{C}_{l_1} (\rho, \{ \ket{0}, \ket{1} \}) = |\sin{\theta}| > 0 \), as illustrated in Fig.~\ref{fig1}. 

Although coherence in the basis \( \{ \ket{a_i} \} \) can be detected via the nonpositivity of the extended KD distribution using a mutually unbiased basis \( \{ \ket{b_k} \} \), it is important to note that the KD distribution itself depends on the choice of \( \{ \ket{b_k} \} \). Consequently, different parameter choices ({\it e.g.}, values of \( \alpha \) and \( \beta \)) may require the evaluation of determinants of Hankel matrices of different orders to detect coherence using our moment-based approach. In Fig.~\ref{fig1}, we consider two different pairs \( (\alpha, \beta) \) and plot the negative of the determinant corresponding to the {\it minimum} order of the Hankel matrix required to detect coherence across the full range of \( \theta \).

\begin{figure}
    \centering
    \includegraphics[width=8cm]{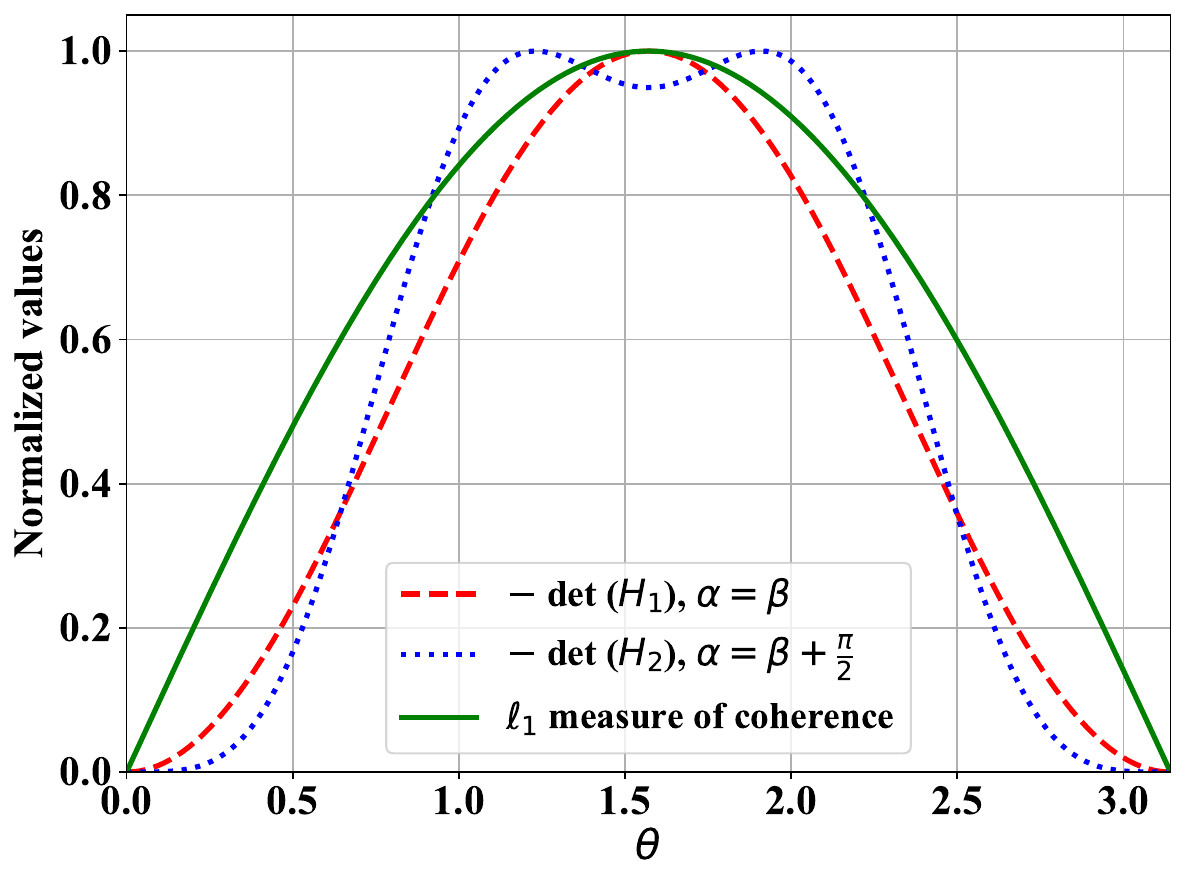}
\caption{
    \justifying 
    \small{
    The plot compares three quantities as functions of $\theta \in [0, \pi]$: 
    (i) the $\ell_1$ measure of coherence (solid green curve), which is exactly equal to the total nonpositivity of the extended KD distribution , 
    (ii) $-\det(H_1)$ for $\alpha = \beta$, (dashed red curve) showing positivity across the entire range (except $\theta =0,\pi$) and thereby detecting coherence in the entire range of $\theta$ at the first level of the hierarchy, and 
    (iii) $-\det(H_2)$ for $\alpha = \beta + \pi/2$, (dotted blue curve) whose positivity demonstrates detection of coherence at the second level of the hierarchy in Theorem \ref{theorem4}.}
    }
    \label{fig1}
\end{figure}

\subsection{Detecting nonclassical extractable work}
We now show the implications of our moment-based criteria in the extraction of nonclassical work.

The main focus of thermodynamics is to understand how the entropy or energy of a system changes due to the interaction with an external environment, such as a heat bath \cite{lostaglio2019introductory}. However, in most of the cases, the thermodynamic quantities of interest, such as heat and work, are exchanged in small amounts, often comparable to their averages, satisfying the fluctuation theorems in classical thermodynamics \cite{PhysRevE.56.5018,crooks1998nonequilibrium}.
One possible approach to address the work extraction protocol is the two-point measurement (TPM) scheme, outlined below \cite{ArvidssonShukur2020}.

Consider the evolution of a thermodynamic system $\rho \propto e^{-\beta \mathcal{H}(t_0)}$ under a Hamiltonian $\mathcal{H}(t)$ having spectral decomposition $\mathcal{H}(t) = \sum_i E_i(t) \Pi_i(t)$ and the energy eigen values $E_i(t)$. A measurement of energy at time $t=0$, gives the $i$th energy eigen value $E_i(t_0)$ and the state is updated to $\frac{\Pi_i(t_0)}{\Tr{\Pi_i(t_0)}}$. The system evolves from $t=0$ to $t=T$ under a unitary $U=e^{-i \int_{0}^{T} \mathcal{H}(t_0)dt/ \hbar}$ and a measurement of energy at $t=T$ results in an outcome $E_j(T)$ and the updated state is given by $\frac{\Pi_j(T)}{\Tr{\Pi_j(T)}}$. The joint probability of obtaining the energy values $E_i(t_0)$ and $E_j(T)$ is given by \cite{arvidsson2024properties}
\begin{equation}
    p_{ij}(\rho)=\Tr{{U^{\dagger}} \Pi_j(T) U \Pi_i(t_0) \rho \Pi_i(t_0)}. \label{jointprob}
\end{equation}

Microscopically, if energy $W$ is exchanged with probability $p(W)$, then the probability density is given by
\begin{equation}
    p(W) = \sum_{ij}p_{ij}(\rho)  \delta(W-[E_i(t_0)-E_j(T)])  \label{probdensity}
\end{equation}
where $\delta (\cdot)$ is the Dirac-$\delta$ function.
However, coherence in the state $\rho$ with respect to the initial Hamiltonian eigenbasis $\mathcal{H}(t_0)$ 
 results in a deviation of the marginal probabilities $\sum_{j} p_{ij} (\rho)$ from its expected value of $\Tr[\Pi_j(T) U \rho U^{\dagger} ]$. This is associated with the inherent disturbance caused by measurement due to the noncommutativity of $\rho$, $\mathcal{H}(t_0)$, and $\mathcal{H}(T)$. Nevertheless, this can be resolved by introducing quasi-probabilities, or KD elements $Q_{ij}$ such that 
\begin{equation}
    Q_{ij} (\rho) = \Tr[\Pi_{j}^{\mathcal{H}}(T)\Pi_{i}(t_0) \rho]  \label{quasiprob}
\end{equation}
where $\Pi_{j}^{\mathcal{H}}(T) = U^{\dagger} \Pi_{j}(T) U$ represents the evolution of the operator $\Pi_{j}(T)$ in the Heisenberg picture. As mentioned earlier, this quasi-probabilities $Q_{ij}(\rho)$ have the properties: $\sum_{ij} Q_{ij} (\rho) = 1$,  $\sum_{i} Q_{ij} (\rho) = \Tr[\Pi_{j}(T) U \rho U^{\dagger}]$ and $\sum_{j} Q_{ij} (\rho) = \Tr[\Pi_{i}(t_0) \rho]$. Using these quasi-probabilities, Eq.~\eqref{probdensity} gets modified to 
\begin{equation}
     \tilde{p}(W) = \sum_{ij}Q_{ij}(\rho)  \delta(W-[E_i(t_0)-E_j(T)]) . \label{quasiprobdensity}
\end{equation}

Further, A valid work distribution should yield an average work value: $\langle W \rangle := \Tr[\mathcal{H}(T) U \rho U^{\dagger}] -\Tr[\mathcal{H}(t_0) \rho]$. While the two-point measurement distribution does not always satisfy this condition {\it i.e.}, $\int p(W) dW \ne \langle W \rangle$, the KD-based distribution does: $\int \tilde{p}(W) dW = \langle W \rangle$. Notably, the KD average and the KD variance differs from their two-point-measurement average and the two-point-measurement variance respectively, whenever $\rho$ exhibits coherence in the energy basis. One may note that to characterize such work-distribution criteria, the real part of the KD distribution---{\it i.e.}, the  MHQ distribution \cite{Margenau61}---can serve as an indicator of nonclassicality, which is quantified by its negativity.

If $Q^{MHQ}_{ij} (\rho)$ are the elements of MHQ distribution, which is defined as in Eq.~\eqref{mhq} by considering the real parts of the KD elements, then the negativity measure associated with the distribution is given by 
\begin{equation}
    \mathcal{N} (Q^{MHQ}(\rho))= -1 + \sum_{ij} |Q^{MHQ}_{ij}(\rho)|.  \label{negativity}
\end{equation}

For a classical probability distribution, $\sum_{ij}|Q^{MHQ}(\rho)| = 1$ and hence, $\mathcal{N} (Q^{MHQ}(\rho)) = 0$. Therefore a value of $\mathcal{N} (Q^{MHQ} (\rho)) > 0$ indicates nonclassicality in the corresponding MHQ distribution, and this can be interpreted as a resource for enhanced work extraction \cite{PRXQuantum.1.010309,PhysRevResearch.6.023280}.

Let us now show the efficacy of our proposed moment-based criteria in detecting this nonclassical extractable work. For completeness, below we define the relevant moments and then present our moment-based criteria corresponding to the MHQ distribution. 
\begin{definition}
The $n$-th order moments corresponding to the MHQ distribution on a state $\rho$ are defined as 
\begin{equation}
    s_{n} = \sum_{ij} (Q^{MHQ}_{ij} (\rho))^n .\label{MHQmoments}
\end{equation}
\end{definition}
\begin{corollary} \label{theorem5}
 If the KD elements corresponding to the MHQ distribution for a quantum state $\rho$, denoted as $Q^{MHQ}_{ij}(\rho)$ are all positive, then, 
 \begin{equation}
     det[H_{m}(\mathbf{s})] \ge 0
 \end{equation}
 where $[H_{m}(\mathbf{s})]_{ij} = s_{i+j+1}$, for $i,j \in \{0,1,\dots ,m\}, \;m \in \mathbb{N}$ and $\mathbf{s} = (s_1, s_2,\dots,s_{2m+1})$ are the corresponding moments defined in Eq. \eqref{MHQmoments}.
\end{corollary}

Now, for $det(H_2) < 0$ , we have $\mathcal{N}(Q^{MHQ}(\rho)) > 0$, indicating the potentiality of our criteria in detecting nonclassical work extraction. This is illustrated below through an explicit example taken from \cite{PRXQuantum.5.030201}.
\begin{example}
Consider the detection of nonclassical work exerted by a qubit. The qubit undergoes unitary evolution ($U$) under a time-dependent Hamiltonian without coupling to any external bath, so that the full internal energy change can be interpreted as work.
Further, consider a time-dependent qubit Hamiltonian of the form:
\begin{equation}
{\mathcal H}(t)=\frac{1}{2} \left[ \Omega \left( \cos\omega t \, \sigma^x + \sin\omega t \, \sigma^y \right) + \omega\sigma^z \right],    
\end{equation}
which physically represents a spin-1/2 particle subject to an effective magnetic field rotating about the \( z \)-axis. Transforming to the rotating frame via the unitary \( U_{\rm rot} = e^{i \omega \sigma^z t/2} \), the Hamiltonian becomes time-independent:
\begin{equation}
    \tilde {\mathcal H} = U_{\rm rot} {\mathcal H} U_{\rm rot}^\dagger +i \, \dot U_{\rm rot} U_{\rm rot}^\dagger = \frac{\Omega}{2} \sigma^x \,,
\end{equation}
so that the total unitary evolution operator becomes
\begin{equation}\label{eq:example_qubit-unitary}
    U(t)=e^{-i \omega \sigma^z t/2}e^{-i \Omega \sigma^x t/2}.
\end{equation} 

To compute the work distribution between \( t= 0 \) and \( t = T \), we use the spectral decomposition of the Hamiltonian:
\begin{align}
    {\mathcal H}(t) = \sum_{\gamma=0,1} \; E_\gamma \Pi_\gamma(t),
 \;  \;  \text{with } \; E_\gamma =\frac{ (-1)^{\gamma+1} \Delta }{2} \label{eq:qubit_energies}
\end{align}
and \( \Delta = \sqrt{\omega^2 + \Omega^2} \). 
The corresponding eigenvectors are given by,
\begin{align}
\Pi_\gamma(t) = \frac{\mathbb{I}}{2} + (-1)^{(\gamma+1)} \ \frac{\Omega(\sigma^x\cos\omega t+\sigma^y\sin\omega t)+\omega\sigma^z}{2\Delta}.\label{eq:qubit_projectors}
\end{align}

The initial state of the qubit is taken to be coherent in the eigenbasis of \( \mathcal{H}(t_0) \), and is parametrized as:
\begin{equation}\label{eq:initial_state_qubit}
    \rho=\begin{pmatrix}
        \Gamma & \xi \\ 
        \xi & 1-\Gamma
    \end{pmatrix},
\end{equation}
where \( 0 \leq \Gamma \leq 1 \) represents the ground state population, and \( \xi \) encodes coherence.

\begin{figure}[ht]
    \centering
    \includegraphics[width=8cm]{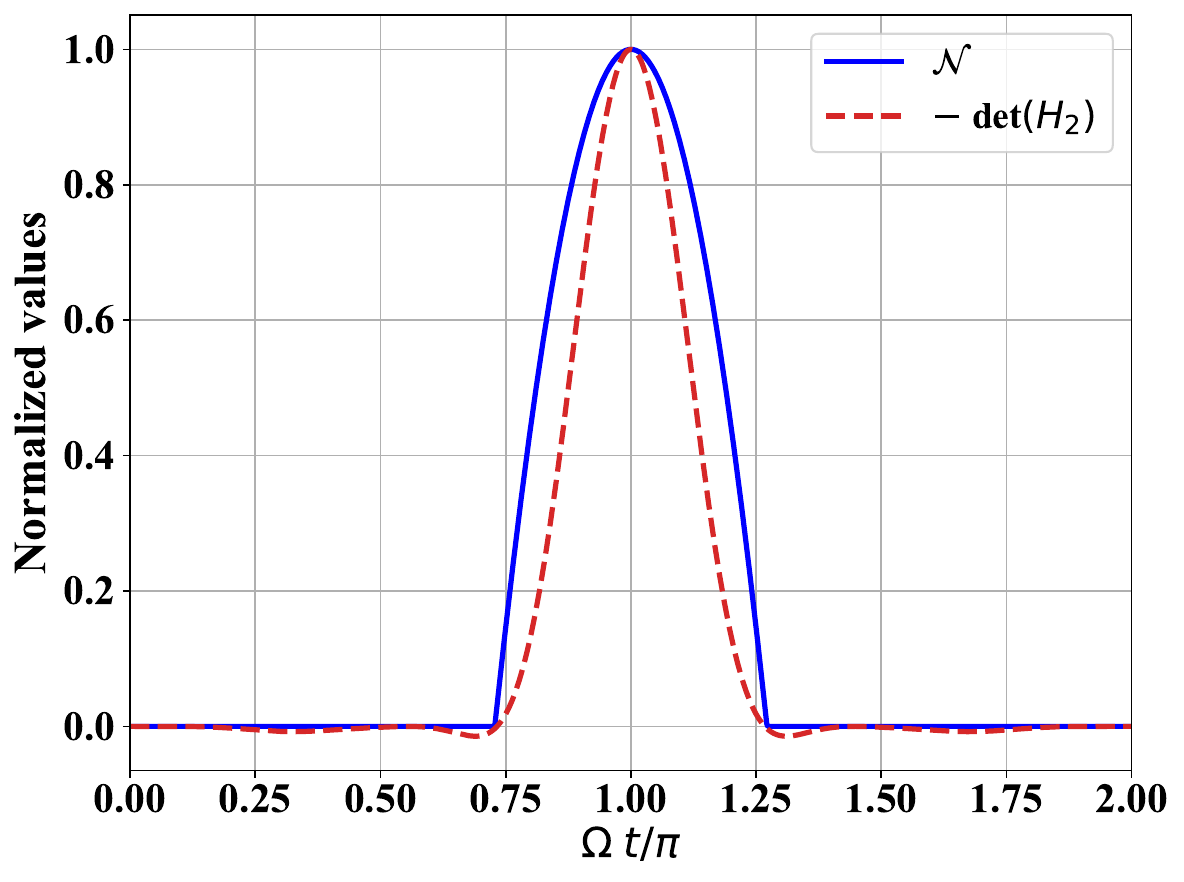}
    \caption{
        \justifying
        \small{
        Plot of the nonpositivity measure \( \mathcal{N}(Q^{MHQ}_{ij}(\rho)) \) as a function of \( \Omega \) (solid blue curve). The regions where \( \mathcal{N} > 0 \) correspond to intervals where nonclassical work extraction is possible. The nonclassical domain is completely identified via Corollary~\ref{theorem5}, based on the determinant of a second-order Hankel matrix constructed from KD moments (dashed red curve).
        }
    }
    \label{fig2}
\end{figure}

Here, we focus on the MHQ distribution. To investigate this feature explicitly, we consider a maximally coherent initial state by choosing \( \Gamma = \xi = 1/2 \). The corresponding MHQ distribution reads:

\begin{eqnarray}
   Q^{MHQ}_{00} &=& \frac{\omega^2 - \omega \Omega + 2\Omega^2 + \omega(\omega + \Omega)\cos\Omega t}{4 \Delta^2}, \nonumber
   \label{eq:Req00}
   \\
   Q^{MHQ}_{01} &=& \frac{\omega(\omega + \Omega)(1 - \cos\Omega t)}{4 \Delta^2}, \nonumber
   \\
   Q^{MHQ}_{10} &=& \frac{\omega(\omega- \Omega)(1 - \cos\Omega t)}{4 \Delta^2}, \nonumber
   \\
   Q^{MHQ}_{11} &=& \frac{\omega^2 + \omega\Omega + 2\Omega^2 + \omega(\omega - \Omega)\cos\Omega t}{4 \Delta^2}. \nonumber
   \label{eq:Req11}
\end{eqnarray}
\end{example}

Now, the MHQ work distribution (following Eq.~\eqref{quasiprobdensity}), is given by:
\begin{eqnarray}
    \tilde{p}(W) &=& \left( Q^{MHQ}_{00} + Q^{MHQ}_{11} \right)\delta(W)+
    \nonumber\\
    &+& Q^{MHQ}_{10}\delta(W+\Delta) + Q^{MHQ}_{01}\delta(W-\Delta).
\end{eqnarray}

As mentioned previously, $\mathcal{N}(Q^{MHQ}(\rho)) >0$ indicates the enhanced work extraction. This behavior is illustrated in Fig.~\ref{fig2}, where the nonpositive regions are clearly shown. Notably, the application of Corollary~\ref{theorem5} using a second-order Hankel matrix captures the full extent of the nonclassical regime, validating our criteria.
\section{Proposal for realization of KD moments using shadow tomography}\label{s5}
In Sec.~\ref{s3}, we have established that the positivity of the KD distribution implies the positivity of the associated Hankel matrices of moments. Consequently, a negative determinant in any of these Hankel matrices serves as a sufficient indicator of nonpositivity. A natural question then arises: \emph{how can one access such KD moments in an experimental setting?} In this section, we discuss an efficient approach based on the framework of shadow tomography.

Shadow tomography, originally introduced by Aaronson~\cite{aaronson2018shadow}, provides a method to estimate a very large (even exponential) number of target functions of a quantum state using only a polynomial number of copies of the state. While highly sample-efficient, Aaronson’s protocol places heavy demands on quantum hardware: in principle, it requires collective operations on all stored copies of the state, implemented through exponentially long quantum circuits. A more practical approach was later developed in the form of \emph{classical shadows}, as introduced by Huang et. al.~\cite{huang2020predicting} with the same objective of predicting the target functions, and not the full state.  Using classical shadows, only $O (\log M)$ measurements are required to accurately estimate $M$ different target functions of the state with high success probability \cite{huang2020predicting}.

A classical shadow $(\hat{\rho})$ of an unknown quantum state $(\rho)$ is created by repeatedly performing a simple procedure: a random unitary is first applied $(\rho \rightarrow U\rho U^\dagger)$, rotating the state, where the transformation $U$ is randomly selected from an fixed ensemble of unitaries, with different ensembles leading to different versions of the procedure; then a computational basis measurement is performed, followed by classical post processing. In each run, a classical snapshot of $\rho$ is obtained, which we write as $\hat{\rho}$. The number of times this procedure is repeated is called the `size' of the classical shadow. Finally, the classical shadow of size $N$ of an unknown quantum state $\rho$ is written as,
\begin{equation}
    S(\rho; N) = \{ \hat{\rho}_1, \hat{\rho}_2, \dots, \hat{\rho}_N \}
\end{equation}
The classical snapshots exactly reproduce the underlying state in expectation over the unitaries and measurement outcomes: $\mathbb{E} (\hat{\rho}) = \rho$. From such shadows, one can efficiently predict many different linear and nonlinear functions with high probability, without the need for full tomography of $\rho$. To connect this framework to KD moments, we note that the $n$-th KD moment defined in \eqref{kdmoments} can be expressed as,
\begin{align}
    q_n &= \sum_{i,j} \left[Q_{ij} (\rho)\right]^n \nonumber \\
    &= \sum_{i,j} \left[ \mathrm{Tr}\left[\Pi_j^f \Pi_i^a \, \rho \right] \right]^n \nonumber \\
     &= \sum_{i,j} \mathrm{Tr}\left[\left(\Pi_j^f \Pi_i^a \right)^{\otimes n} \, \rho^{\otimes n}\right] \nonumber \\
     &= \mathrm{Tr} \left[ \left( \sum_{i,j} \left(\Pi_j^f \Pi_i^a \right)^{\otimes n} \right)  \rho^{\otimes n} \right] \nonumber \\
     & = \mathrm{Tr}\left[ B_n \rho^{\otimes n} \right] \;\; \;\text{where } B_n = \sum_{i,j} \left(\Pi_j^f \Pi_i^a \right)^{\otimes n}
\end{align}
Here $\Pi_j^f$ and $\Pi_i^a$ are the projectors associated with two measurement bases $\{ \ket{a_i} \}$ and $\{ \ket{f_j} \}$, respectively. The operator $B_n$ acts on $n$ copies of the state, making $q_n$ a nonlinear function of $\rho$.

It is important to emphasize that, for a fixed quantum state, different choices of measurement bases generally yield different KD distributions, all representing the same state in distinct quasiprobability frames. To eliminate this ambiguity, we focus on the operationally relevant case where the measurement bases are predetermined and known, while the state itself is not. Given the known bases, the associated projectors are also fixed. Thus, the construction of the operator $B_n$ constitutes a classical preprocessing step carried out before any experimental run. During the actual experiment, the experimentalist only needs to implement the unitary operation corresponding to $B_n$ and estimate its expectation value from the classical shadows, without requiring explicit access to each individual projector.

As an example, we illustrate the procedure for the realization of the second order KD moment, which can be written as,
\begin{align}
    q_2 &= \mathrm{Tr}[B_2 \rho^{\otimes 2}] \nonumber \\
    &= \mathrm{Tr}[B_2 (\mathbb{E}\hat{\rho}_i) \otimes (\mathbb{E}\hat{\rho}_j)] \nonumber\\
    &= \mathbb{E} (\mathrm{Tr}[B_2 \hat{\rho}_i \otimes \hat{\rho}_j])
\end{align}
$q_2$ is a quadratic function where $B_2$ acts on two copies of the state. Thus, one can predict $q_2$ using two independent snapshots $\hat{\rho}_i, \hat{\rho}_j, i\neq j$, and $\mathrm{Tr}[B_2 \hat{\rho}_i \otimes \hat{\rho}_j]$ correctly predicts $q_2$ \cite{huang2020predicting}. Similarly, any higher order KD moment $q_n$ can be predicted with expectation of $n$ independent classical snapshots.

This approach offers two key advantages. Firstly, the number of required measurements scales only polylogarithmically in the number of observables being estimated, making it more efficient than brute-force tomography when the task is solely to detect KD nonpositivity\cite{huang2020predicting}. Secondly, since the KD moments are global quantities built from all $Q_{ij}$, shadow tomography provides a natural way to access them without explicitly reconstructing the full KD distribution. Thus, by applying shadow tomography techniques, one can experimentally detect nonpositivity of the KD distribution, and thereby detect quantum resources inherently linked to it, as described in Sec.~\ref{s4}, using only a modest number of measurements.

\section{Conclusions}\label{s6}
Kirkwood-Dirac (KD) quasiprobability distribution has gained renewed attention as a versatile framework finding applications across a wide spectrum of quantum information tasks, including the analysis of weak measurements \cite{Dressel15, Yunger18}, translational asymmetry \cite{Budiyono_2023_asymmetry, budiyono_2023_operational}, quantum metrology \cite{ArvidssonShukur2020}, quantum thermodynamics \cite{lostaglio2020certifying,Upadhyaya24}, as well as quantum foundations \cite{kunjwal2019anomalous, pusey2014anomalous, dressel2011experimental}. Efficiently detecting their nonclassical features is therefore essential, both theoretically and experimentally. In this work, we introduced a moment-based approach to detect nonpositivity of the KD distribution. In many scenarios, the first three moments suffice to reveal KD nonpositivity; however, we also provided examples where higher-order moments are necessary for the detection of KD nonpositivity. Moreover, we extended this moment-based framework to detect quantum resources---such as quantum coherence and nonclassical extractable work---that are fundamentally linked to the KD distribution, thereby illustrating the broader utility of our approach. Finally, we proposed a methodology based on shadow tomography for efficiently computing KD moments in real experiments, thereby demonstrating the practical feasibility of our protocol.

While several existing works have explored the detection of nonpositive KD distributions \cite{gonzalez2019out,PhysRevA.109.012215,PhysRevA.76.012119,PhysRevLett.110.230602,PhysRevLett.110.230601,lundeen2011direct,PhysRevLett.108.070402,buscemi2013direct,tan2024}, our approach offers a distinct advantage in terms of practical implementation. Unlike witness-based \cite{langrenez2024convex} and cloning-based approaches \cite{buscemi2013direct}, moment-based detection scheme does not require prior knowledge of the distribution and is significantly more resource-efficient. Moreover, compared to PVM-based approaches \cite{budiyono2024sufficient} that require optimization over the full set of bases and weak-value–based schemes that become challenging for large or noisy systems, our moment-based approach based on shadow tomography involves only simple functionals and a polylogarithmic number of state copies, ensuring both efficiency and experimental feasibility. Other existing methods to detect KD nonpositivity often rely on full state tomography or direct reconstruction \cite{PhysRevLett.110.230602, PhysRevLett.110.230601,lundeen2011direct,PhysRevLett.108.070402}, which require estimating the entire KD distribution of a quantum state in finite dimension $d$ to precision $\epsilon$, using $O(d^4/{{\epsilon}^2})$ samples. The quantum circuit introduced in \cite{wagner2024quantum} offers an improvement in sample complexity compared to full state tomography, reducing the requirement to $O (d^2/{{\epsilon}^2})$ samples. In contrast, our approach for detecting KD nonpositivity is fundamentally based on shadow tomography \cite{PhysRevLett.125.200501,aaronson2018shadow, huang2020predicting}, which enables the estimation of the desired functions using only a polylogarithmic number of samples. 

Apart from this, our study opens up several other promising directions for future research. While we have focused on detecting nonclassicality via negativity in quasiprobability distributions, nonclassical features can also arise in positive distributions, such as the Husimi Q-function \cite{husimi1940some}. Therefore, a potential direction is to develop moment-based criteria capable of detecting nonclassicality for such distributions. Moreover, considering the practical feasibility of our proposed moment-based approach, a key future direction is its realization within real experimental frameworks.

\section{Acknowledgements}\label{s7}
We thank Nirman Ganguly for useful discussions and insightful comments. B.M. acknowledges the DST INSPIRE fellowship program for financial support. AGM acknowledges the financial support through the National Quantum Mission (NQM) of the Department of Science and Technology, Government of India.
\bibliography{main}

\end{document}